\documentclass[11pt]{article}
\usepackage[letterpaper, portrait, margin=1in]{geometry}
\usepackage[hypertexnames=false,colorlinks=true,linkcolor=blue,citecolor=ForestGreen]{hyperref}
\usepackage{algorithm}
\usepackage[noend]{algpseudocode}
\usepackage{url}
\usepackage{amsmath,amssymb,amsthm}
\usepackage{thmtools,thm-restate}
\usepackage[noabbrev,capitalise,nameinlink]{cleveref}
\usepackage{mathtools}
\usepackage{xspace}
\usepackage{verbatim}
\usepackage{mathrsfs}
\usepackage[usenames,dvipsnames,svgnames,table]{xcolor}
\usepackage{pgf}
\usepackage[dvipsnames]{xcolor}
\usepackage{todo}
\usepackage{tabularx}
\usepackage{enumitem}
\usepackage{derivative}
\usepackage{bm}
\usepackage{multirow}
\usepackage{diagbox}
\usepackage{nicematrix}
\usepackage{parskip}
\usepackage{adjustbox}
\usepackage{tikz}
\usepackage{tikz-3dplot}
\usepackage{transparent}
\usepackage{subcaption}
\usepackage[most]{tcolorbox}

\definecolor{theoremcolor}{HTML}{ffefec} 
\definecolor{resultcolor}{HTML}{FFE1D9}
\definecolor{citecolor}{HTML}{35641a}
\definecolor{linkcolor}{HTML}{d30045} 

\hypersetup{
    colorlinks=True,
    citecolor=citecolor,
    linkcolor=linkcolor,
}

\usepackage{dsfont}

\newcommand*\ie{i.\kern.1em e., }
\newcommand*\eg{e.\kern.1em g., }
\newcommand*\cf{c.\kern.1em f.\ }
\newcommand*\almev{a.\kern.1em e.\ }
\newcommand*\iid{i.i.d.\ }

\theoremstyle{plain}
\newtheorem{theorem}{Theorem}[section]
\newtheorem{lemma}[theorem]{Lemma}

\newtheorem{proposition}[theorem]{Proposition}

\newtheorem{corollary}[theorem]{Corollary}
\newtheorem{question}[theorem]{Question}

\newtheorem{observation}[theorem]{Observation}
\newtheorem{assumption}[theorem]{Assumption}
\newtheorem{definition}[theorem]{Definition}

\crefname{claim}{Claim}{Claims}
\crefname{fact}{Fact}{Facts}

\theoremstyle{definition}

\newtheorem{remark}[theorem]{Remark}
\newtheorem{example}[theorem]{Example}

\theoremstyle{plain}

\newenvironment{boxtheorem}{\begin{theorem}}{\end{theorem}}
\tcolorboxenvironment{boxtheorem}{colback=theoremcolor, colframe=white,
    colbacktitle=theoremcolor, coltitle=theoremcolor}

\tcolorboxenvironment{boxlemma}{colback=resultcolor, colframe=white,
    colbacktitle=resultcolor, coltitle=resultcolor}

\tcolorboxenvironment{boxproposition}{colback=resultcolor, colframe=white,
    colbacktitle=resultcolor, coltitle=resultcolor}

\tcolorboxenvironment{boxcorollary}{colback=resultcolor, colframe=white,
    colbacktitle=resultcolor, coltitle=resultcolor}

\tcolorboxenvironment{boxobservation}{colback=resultcolor, colframe=white,
    colbacktitle=resultcolor, coltitle=resultcolor}

\tcolorboxenvironment{boxquestion}{colback=questioncolor, colframe=white,
    colbacktitle=questioncolor, coltitle=shin-kai}
\newenvironment{boxquestion*}{\begin{question*}}{\end{question*}}
\tcolorboxenvironment{boxquestion*}{colback=questioncolor, colframe=white,
    colbacktitle=questioncolor, coltitle=questioncolor}

\tcolorboxenvironment{boxdefinition}{colback=definitioncolor, colframe=white,
    colbacktitle=definitioncolor, coltitle=definitioncolor}

\tcolorboxenvironment{boxassumption}{colback=definitioncolor, colframe=white,
    colbacktitle=definitioncolor, coltitle=definitioncolor}

\tcolorboxenvironment{boxexample}{colback=examplecolor, colframe=white,
    colbacktitle=examplecolor, coltitle=examplecolor}

\newtheorem{exercise}[theorem]{Exercise}

\tcolorboxenvironment{boxexercise}{colback=exercisecolor, colframe=white,
    colbacktitle=exercisecolor, coltitle=exercisecolor}

\tcolorboxenvironment{boxgoal}{colback=goalcolor, colframe=white,
    colbacktitle=goalcolor, coltitle=goalcolor}

\newenvironment{boxgoal*}{\begin{goal*}}{\end{goal*}}
\tcolorboxenvironment{boxgoal*}{colback=goalcolor, colframe=white,
    colbacktitle=goalcolor, coltitle=goalcolor}

\tcolorboxenvironment{boxremark}{colback=remarkcolor, colframe=white,
    colbacktitle=remarkcolor, coltitle=remarkcolor}

\newcommand{\ignore}[1]{}

\DeclareMathOperator{\supp}{supp}   

\DeclareMathOperator{\sign}{sign}

\newcommand{\dist}{\mathsf{dist}}

\newcommand{\Accept}{\mathsf{Accept}}
\newcommand{\Reject}{\mathsf{Reject}}

\newcommand{\Ex}[1]{\bE \left[ #1 \right]}
\newcommand{\Exu}[2]{\underset{#1} \bE \left[ #2 \right] }
\newcommand{\Exuc}[3]{\underset{#1} \bE \left[ #2 \;\; \left| \;\; #3
\right.\right] }

\renewcommand{\Pr}[1]{\bP \left[ #1 \right]} 
\newcommand{\Pru}[2]{\underset{ #1 }\bP \left[ #2 \right]}

\newcommand{\define}{\vcentcolon=}

\renewcommand{\epsilon}{\varepsilon}

\newcommand{\floor}[1]{\ensuremath{\left\lfloor #1 \right\rfloor}}

\newcommand{\inn}[1]{\langle #1 \rangle}

\newcommand{\ind}[1]{\mathds{1} \left[ #1 \right] }

\newcommand{\pmset}{\{\pm 1\}}

\newcommand{\cC}{\ensuremath{\mathcal{C}}}
\newcommand{\cD}{\ensuremath{\mathcal{D}}}

\newcommand{\cH}{\ensuremath{\mathcal{H}}}

\newcommand{\cN}{\ensuremath{\mathcal{N}}}

\newcommand{\cS}{\ensuremath{\mathcal{S}}}

\newcommand{\cX}{\ensuremath{\mathcal{X}}}

\newcommand{\bE}{\ensuremath{\mathbb{E}}}

\newcommand{\bP}{\ensuremath{\mathbb{P}}}
\newcommand{\bR}{\ensuremath{\mathbb{R}}}
\newcommand{\bS}{\ensuremath{\mathbb{S}}}

\newcommand{\ba}{\boldsymbol{a}}
\newcommand{\bb}{{\boldsymbol{b}}}

\newcommand{\bs}{\boldsymbol{s}}

\newcommand{\bw}{\boldsymbol{w}}
\newcommand{\bx}{\boldsymbol{x}}

\newcommand{\bz}{\boldsymbol{z}}

\newcommand{\abs}[1]{{\left| #1 \right|}}
\newcommand{\vabs}[1]{{\left\| #1 \right\|}}
\newcommand{\abra}[1]{{\left\langle #1 \right\rangle}}
\newcommand{\pbra}[1]{{\left( #1 \right)}}

\newcommand{\TV}{\mathsf{TV}}

\newcommand{\yes}{{\mathrm{yes}}}
\newcommand{\no}{{\mathrm{no}}}
\newcommand{\Dyes}{\cD_{\yes}}
\newcommand{\Dno}{\cD_{\no}}
\newcommand{\bdelta}{\boldsymbol{\Delta}}
\newcommand{\bX}{\boldsymbol{X}}
\newcommand{\TupleS}{\bm{\cS}}
\newcommand{\btheta}{\bm{\theta}}
\newcommand{\Bern}{\mathrm{Bern}}

\newcommand{\Lock}{\mathrm{Lock}}
\newcommand{\Unlock}{\mathrm{Unlock}}

\title{Distribution-Free Halfspace Testing with Samples}
\author{Xi Chen \\ Columbia University
\and
Renato Ferreira Pinto Jr. \\ Columbia University
\and
Nathaniel Harms \\ University of British Columbia
\and
Shyamal Patel \\ Columbia University 
\and
Rocco A.~Servedio \\ Columbia University}

\date{}

\begin{document}

\maketitle

\begin{abstract}
We prove a tight $\Theta(n/\epsilon)$ lower bound on the number of samples required for testing halfspaces
over $\bR^n$, in the distribution-free sample-based model where the underlying probability
distribution is unknown to the algorithm, and the algorithm only receives random samples (\ie it
cannot make queries). This shows that testing is no more efficient than learning for halfspaces. We also show a matching upper bound 
for one-sided testers, improving on the standard (two-sided) testing-by-learning reduction, establishing
that two-sided halfspace testers in this model have no advantage over one-sided testers.
\end{abstract}

\thispagestyle{empty}
\setcounter{page}{0}
{
\setcounter{tocdepth}{2}
\tableofcontents
}

\thispagestyle{empty}
\setcounter{page}{0}
\newpage
\setcounter{page}{1}

\newcommand{\red}[1]{{\color{red} {#1}}}
\newcommand{\blue}[1]{{\color{blue} {#1}}}
\newcommand{\gray}[1]{{\color{gray} {#1}}}
\newcommand{\violet}[1]{{\color{violet} {#1}}}

\newcommand{\shyamal}[1]{\blue{\textbf{Shyamal:} {#1}}}

\newcommand{\rocco}[1]{\violet{\textbf{Rocco:} {#1}}}

\newcommand{\rsnote}[1]{\footnote{{\color{blue}\textbf{Rocco:}} {#1}}}

\allowdisplaybreaks

\section{Introduction}

We would like to understand when the simplest and most trivial property testing algorithm is also
\emph{optimal}, thereby justifying our laziness and ineptitude in algorithm design.

Let us explain what we mean by the simplest and most trivial algorithm.
In this paper, we are interested in \emph{distribution-free sample-based} property testing. To
motivate and define these testers, consider the problem of learning halfspaces, which is one of the
most well-studied problems in learning theory. We receive random samples $(\bm x, f(\bm x))$, where
$\bm x \sim \cD$ is drawn from an arbitrary distribution $\cD$
(unknown to us) and labelled by an unknown function $f$.
The goal is to learn $f$.

The first step is to choose an appropriate \emph{hypothesis class} $\cH$ (\ie a set of functions), which we'd like to be ``simple'' while also containing the unknown $f$.
A common choice, both in theory and in practice, is the set of \emph{halfspaces} (or linear
separators), which are functions of the form
\[
    h(x) = \sign\big(\inn{w,x} - t\big) \qquad\text{where}\qquad w \in \bR^n, t \in \bR .
\]
Then, we collect samples and output some function $h \in \cH$ that agrees with $f$ on the samples.
The beautiful and well-developed VC dimension theory \cite{VC71,VC74}, and its continuing
refinement, tells us that, for our class of halfspaces, if we draw
$\Theta(\tfrac{n+\log(1/\delta)}{\epsilon})$ samples, then with probability at least $1-\delta$, we
can choose (using the SVM rule \cite{VC74,BHMZ20}) a halfspace $h$ that will satisfy
\[
    \epsilon > \dist_\cD(f,h) \define \Pru{\bm x \sim \cD}{f(\bm x) \neq h(\bm x)} ,
\]
\emph{if the unknown function $f$ is indeed consistent with a halfspace,} 
\ie 
$f \in \cH$. But how do we know whether this assumption about $f$ is true? This is the purpose of a
\emph{halfspace tester}: it should output $\Accept$ when indeed $f \in \cH$, and it should output
$\Reject$ when $\dist_\cD(f,h) > \epsilon$ for all halfspaces $h$. The tester is called
\emph{distribution-free} and \emph{sample-based} when it is given the same input as the
learning algorithm above: random samples $(\bm x, f(\bm x))$ from an arbitrary unknown $\cD$.

So, what is the \emph{simplest} distribution-free sample-based tester (or indeed
the simplest tester in any model)? It is a simpler version of the learning
algorithm: Collect samples and output $\Accept$ if there \emph{exists} $h \in
\cH$ that agrees with $f$ on all of the samples; otherwise output $\Reject$.
This trivial algorithm has one-sided error (\ie it never outputs $\Reject$ when
$f \in \cH$), and indeed it is the \emph{unique} one-sided-error tester when
only random samples are provided.

When is this trivial tester the optimal one? While the sample complexity of \emph{learning} has a
well-developed theory, we do not have a good understanding of the sample complexity of
\emph{testing}. While the tester is trivial, the analysis of its sample complexity is not, and it is
not even clear \emph{a priori} whether the trivial tester can match the sample complexity of 
learning (the classic testing-by-learning reduction \cite{GGR98} produces a two-sided
error tester). We know that the trivial tester is sometimes more efficient than learning, in
particular for monotone functions and juntas \cite{BFH21,BHK26}. It is the optimal tester for juntas
\cite{BHK26} but \emph{not} the optimal tester for the class of functions of bounded support
\cite{FH26}.

In this paper we show that the trivial tester is the optimal distribution-free sample-based
halfspace tester, and that testing has no advantage over learning:

\begin{boxtheorem}
    \label{thm:main}
Distribution-free sample-based halfspace testing over $\bR^n$ requires
$\Theta\left(\tfrac{n+\log(1/\delta)}{\epsilon}\right)$ samples. The upper bound is achieved by
the one-sided error tester, and the lower bound holds even for two-sided testers on the restricted domain $\pmset^n$.
\end{boxtheorem}

\paragraph*{Comparison to prior work.} When the distribution
is uniform over an arbitrary point set, \cite{ES20} gave an
$O(n/\epsilon)$-sample two-sided tester, and, more generally, a two-sided error $O(n/\epsilon)$
distribution-free bound follows from the
standard testing-by-learning reduction \cite{GGR98} combined with the relatively recent optimal bounds on (proper\footnote{A \emph{proper} learner for a class $\cH$ must output a function in $\cH$; these are required by the standard testing-by-learning reduction \cite{GGR98}.}) learning of
halfspaces \cite{BHMZ20}. 
More sophisticated sample-based halfspace testers
with better sample complexity (but necessarily with two-sided error) are known
under distributional assumptions, in particular when $\cD$ is Gaussian
\cite{BBBY12} or rotation-invariant \cite{Har19}.  When queries are allowed,
there is a constant-query tester for the Gaussian distribution and for the
uniform distribution over $\pmset^n$ \cite{MORS10}. 

To prove our one-sided error upper bound, we prove a new general bound for
one-sided testers using stable sample-compression schemes (see
\cref{section:ub}), which were previously used by \cite{BHMZ20} to prove the
sample optimality of the Support-Vector Machine (SVM) algorithm for learning
halfspaces. For classes which admit small compression schemes, our result improves upon the standard
testing-by-learning reduction of \cite{GGR98} that gives a two-sided tester, and
upon the $O\left(\tfrac{\mathsf{VC}}{\epsilon}\log(1/\epsilon)\right)$ one-sided
upper bound of \cite{BFH21}, where $\mathsf{VC}$ denotes the VC dimension (for
halfspaces this is $n+1$). 

The lower bound in our theorem improves on a lower bound of
$\Omega\left(\frac{n}{\epsilon \log n}\right)$ proved by a more general method
in \cite{BFH21}, which gave similar lower bounds for other hypothesis classes,
and has also been extended to work against 
testers that can make queries \cite{CP22}. 
\cref{thm:main} also recovers the $\Omega(n/\epsilon)$ lower bound stated in
\cite{ES20}, with applications for testing LP-type problems, whose proof
unfortunately had a mistake (as noted by those authors, see
\cite[Remark 6.11]{BFH21}), and whose correction led to the weaker bound of
\cite{BFH21}.

\paragraph*{Qualitative significance of these improvements.}
We are interested in tightening these prior bounds for halfspace testing for several
reasons. The first reason is that the general lower bound method of \cite{BFH21}
relies on lower bounds for the support-size distinction problem
\cite{RRSS09,VV11,WY19}, which use sophisticated moment-matching techniques; our
new proof is self-contained and elementary. Also, the general method of
\cite{BFH21} cannot be tightened, because it is already tight for some
(non-halfspace) classes \cite{FH26}.

The second reason is that halfspaces are crucial for understanding the \emph{testing vs.~learning}
question of \cite{GGR98}, a central question in property testing
which asks when testing can be done more efficiently than learning. One motivation for this question
is the task described earlier: testing whether a chosen hypothesis class is suitable for learning an
unknown function $f$, while ideally using fewer samples than required for learning. Halfspaces are
arguably the most important class of functions to understand from this perspective: they are among
the most well-studied and widely-used functions in learning theory, so we should aim to understand
testers for them equally well.

The third reason is to better understand one-sided vs.~two-sided error.
Recent work \cite{BHK26} argued that, for sample-based testers, one-sided vs.~two-sided error
is a more interesting comparison than testing vs.~learning. As noted above, this is a comparison 
between trivial and non-trivial algorithms. Also, the general
$O\left(\tfrac{\mathsf{VC}}{\epsilon}\log(1/\epsilon)\right)$ upper bound on one-sided testing
\cite{BFH21} means that the trivial and unique one-sided tester is at most a factor
$\log(1/\epsilon)$ worse than the optimal learning algorithm, but can sometimes be significantly
more efficient, as for monotone functions \cite{BFH21} and juntas \cite{BFH21,BHK26}; see
\cite{Han16,Han19,BHMZ20,AHMLZ24} for discussion of the optimal sample size of learning. So
one-vs.-two-sided error is in some sense a sharper question than testing vs.~learning, and demands
more insight into the algorithmic techniques.

Since every labeling of $n+1$ points in general position can be achieved by a halfspace, a one-sided error tester under the Gaussian distribution requires at least $n+2$ samples. Therefore, any distribution-free halfspace tester using $o(n)$ samples, even if it were only $O(n/\log n)$ samples to match the
previous lower bound, must somehow take advantage of having two-sided error.  This can be done for some problems: the class of functions of support size $\leq k$ (\ie functions $f$ where $|\{ x \colon f(x)
= 1 \}|\leq k$) requires $\Theta(k)$ samples to learn or to test with one-sided error, but only
$\Theta(k / \log k)$ samples to test with two-sided error \cite{GR16,FH26}, thanks to surprising
algorithmic techniques for support-size estimation \cite{VV11,WY19}. 
The impetus for this work was to either find or rule out
similarly surprising techniques for halfspaces;
our lower bound shows that no such techniques exist for halfspaces.
In short, the $\log n$ factor is important for qualitative
understanding of when and how non-trivial testing algorithms can exist.

\section{Lower Bound Proof}
\label{section:main-lb}

We present two proofs of the lower bound:
\begin{flushleft}\begin{enumerate}
\item An intuitive human-generated proof which uses only elementary trigonometry and basic
    concentration of measure, but only works for domain $\bR^n$ instead of $\pmset^n$; and
\item A simple and surprising AI-generated proof for domain $\pmset^n$ that uses a
completely different construction that we did not expect.
\end{enumerate}\end{flushleft}

We present the more intuitive human-generated proof for $\bR^n$ as an ``alternative'' proof in
\cref{section:alternate}. The main body of the paper will present our own (human-written) exposition
of the second proof, because it is stronger to prove the theorem for domain $\pmset^n$ than for
$\bR^n$, and it has the more surprising technique. 

\begin{remark}
\label{remark:ai}
The AI-assisted proof was found in interactions with ChatGPT 5.5 Pro by prompting it to brainstorm
ideas for the problem and then asking it to pursue what it believed was the most promising strategy.
We needed a few further rounds of interaction to correct errors in GPT's proposed proofs. We found
it interesting that the proof was not ``replicable'': when a different author repeated a similar
process, the AI was not able to find the same proof again (though it did find a different but less
interesting and more tedious proof).
\end{remark}

We first give the formal   definition of halfspace testing over the domain $\pmset^n$.

\begin{definition}[Halfspace Testing]
    \label{def:halfspace-testing}
A \emph{distribution-free sample-based} halfspace tester for domain $\pmset^n$ is an algorithm which
takes as input a proximity parameter $\epsilon > 0$ and confidence parameter $\delta > 0$. For any
function $f \colon \pmset^n \to \pmset$ and distribution $\cD$ over $\pmset^n$, the algorithm
requests samples of the form $(\bm x, f(\bm x))$ where each $\bm x \sim \cD$ is drawn independently
from $\cD$. With probability at least $1-\delta$ over the samples and the internal randomness of
the algorithm, the output must satisfy:
\begin{itemize}
\item \emph{Completeness.} If $f$ is a halfspace, the algorithm outputs $\Accept$;
\item \emph{Soundness.} If $\dist_\cD(f,h) > \epsilon$ for all halfspaces $h$, the algorithm
outputs $\Reject$.
\end{itemize}
The sample complexity $m = m(n,\epsilon, \delta)$
is the number of samples drawn by the algorithm in the
worst case.
\end{definition}

\subsection{Lower Bound ``Lock and Key'' Construction}

The proof proceeds by Yao's principle. We construct
meta-distributions $\Dyes$ and $\Dno$ such that:
\begin{flushleft}\begin{enumerate}[itemsep=0pt]
\item A distribution $\cD$ drawn from either $\Dyes$ or $\Dno$ is a distribution over labeled
    samples $(\bm z, \bm b) \in \pmset^n \times \pmset$.
\item Samples $(\bm z, \bm b)$ from distribution $\cD \sim \Dyes$ are labeled according to some
    halfspace (\ie $\bm b = h(\bm z)$ for some halfspace $h$),
    while samples from distribution $\cD \sim \Dno$ are labeled randomly, i.e., under $\cD$ we have that $\bm b$ and $\bm z$ are independent and $\bm b$ is drawn uniformly from $\pmset$,
    and therefore are far from
being consistent with any halfspace.
\item An algorithm that draws $m = o(n)$ samples from the distribution $\cD$ cannot distinguish
    between the cases $\cD \sim \Dyes$ and $\cD \sim \Dno$; specifically, our main lemma is:
\end{enumerate}\end{flushleft}
\begin{lemma}[Indistinguishability]
\label{lemma:indistinguishability}
Let $\TupleS_\yes$ be the distribution of $m$ independent samples drawn from a randomly chosen
distribution $\cD \sim \Dyes$, and let $\TupleS_\no$ be the distribution of $m$ independent samples
drawn from a distribution $\cD \sim \Dno$. Then for $m = o(n)$,
\[
    \dist_\TV(\TupleS_\yes, \TupleS_\no) = o(1) ,
\]
where $\dist_\TV$ denotes the total variation (TV) distance.
\end{lemma}

This lemma is the interesting part of the main lower bound in \cref{thm:main}, but to complete the
proof there are some technicalities that we must handle.  First, as to be defined, samples drawn from
distributions $\cD \sim \Dno$ are labeled randomly, whereas to satisfy \cref{def:halfspace-testing},
we need them to be labeled by a fixed function $f:\{\pm 1\}^n\rightarrow \{\pm 1\}$ that is far from being a halfspace. Second, we
need a further extension to the argument to obtain the dependence on $\epsilon$.  Finally, we need to
incorporate dependence on the parameter $\delta > 0$.

These technicalities are handled by standard techniques. Informally, as long as $m$ samples drawn
from the NO distribution have negligible probability of producing any collisions (\ie two sample
points coming from the same domain element), we can turn random labels into a fixed (but still
random) function $f$. And, to get the dependence on $\epsilon$, we can introduce a ``dummy element''
into the domain with probability mass $1-\epsilon$, and place the remaining $\epsilon$ mass
according to the distributions in \cref{lemma:indistinguishability}. Finally, we can prove an
$\Omega \left(\frac{\log(1/\delta)}{\epsilon} \right)$ bound by showing any tester requires $\Omega(\log(1/\delta))$ samples to distinguish the parity function $x_1x_2$ and the dictator $x_1$ with probability $1 - \delta$ and applying the previous dummy element argument to again get the desired dependence on $\epsilon$. These technicalities are
handled formally in \cref{section:technicalities}, and for now we focus only on proving
\cref{lemma:indistinguishability}.

Let us begin by defining how we construct our hard distributions $\Dyes$ and
$\Dno$, and then we will give a proof overview and full proof of
\cref{lemma:indistinguishability} in \cref{section:indistinguishability}.

\subsubsection*{Lock and Key Construction}

For notational convenience, we will construct distributions over domain $\pmset^{n + \ell}$ instead
of $\pmset^n$; we will take $\ell = \ell(n) = \lceil 10 \log n \rceil$, so this change of variables
does not affect the bounds.

Our meta-distributions $\Dyes$ and $\Dno$ are both parameterized by a vector $\bm a \sim [0,B)^n$,
drawn uniformly at random from a large box with $B = n^3$, which we call a ``\emph{key}''\!. Given any  fixed
key $a \in [0,B)^n$, we define distributions $\Dyes^a$ and $\Dno^a$ over $\{\pm 1\}^n\times \{\pm 1\}$ below.
This then finishes the construction of $\Dyes$ and $\Dno$: to draw $\cD\sim \Dyes$ (or $\cD\sim \Dno$), one first draws a key $\ba\sim [0,B)^n$ and then sets $\cD$ to be $\Dyes^{\bm a}$ (or $\Dno^{\bm a}$, respectively).

\subsubsection*{The Distribution $\Dyes^a$}

To draw a sample $(\bm z,\bm b)\sim \Dyes^a$, we start with a uniformly random hypercube point $\bm x
\sim \pmset^n$ and a uniformly random label $\bm b\sim \{\pm 1\}$. 
We then use $\bm x$ and $\bm b$ to obtain
\[
    \bm z = \Lock_a(\bx, \bb) \in \pmset^{n + \ell} ,
\]
where the ``\emph{locking mechanism}'' is defined as
\begin{equation}
    \label{eq:locking-mechanism}
    \Lock_a(x, b) \define x \circ \left\langle \lfloor a^\top x \rfloor + b \right\rangle
\end{equation}
with $\circ$ denoting concatenation, and $\langle k \rangle \in \pmset^\ell$ denoting the binary
representation in two's complement with $-1$s in place of $0$s
for the integer $k \in [-2^{\ell-1}, 2^{\ell-1} -1]$, \ie the string $w \in \{\pm 1\}^\ell$ corresponds to $-2^{\ell-1} \cdot (1 + w_1)/2 + \sum_{i=2}^\ell 2^{\ell-i} \cdot (1 + w_i)/2$.
When $B = n^3$, it will suffice to use $\ell = O(\log n)$ bits to represent $\lfloor a^\top x \rfloor + b$.
This finishes the description of $\Dyes^a$.

The reason that samples $(\bm z,\bm b)\sim \Dyes^a$ are labeled according to a halfspace is that,  using the ``key'' $a$, we can define a halfspace $h_a$ that
``\emph{unlocks}'' the mechanism to recover the value $b$:
\begin{proposition}
    \label{lemma:yes-halfspace}
    For each $a \in [0, B)^n$, there exists a halfspace $h_a : \pmset^{n+\ell} \to \pmset$ such that
    \[
        h_a\big(\Lock_a(x, b)\big) = b, \quad \text{for each $x \in \pmset^n$ and $b \in \pmset$.}
    \]
\end{proposition}
\begin{proof}
    The affine function $\Unlock_a : \pmset^{n+\ell} \to \bR$ given by
    \begin{equation}\label{eq:hehe1}
        \Unlock_a(z) \define \sum_{i=1}^n (-a_i) z_i - 2^{\ell-1} \pbra{\frac{z_{n+1}+1}{2}}
            + \sum_{i=2}^{\ell} 2^{\ell-i} \pbra{\frac{z_{n+i}+1}{2}}
    \end{equation}
    satisfies $$\Unlock_a(z) = -a^\top x + \lfloor a^\top x\rfloor + b$$ for each $z = \Lock_a(x, b)$ 
    (the
    first term of \Cref{eq:hehe1} computes $-a^\top x$ and the remaining terms compute 
    $\lfloor a^\top x\rfloor+b$).
    Since $\abs{\lfloor a^\top x\rfloor - a^\top x} < 1$ while $\abs{b} = 1$, we always
    have $\sign(\Unlock_a(z)) = b$, and hence the halfspace $h_a(z) \define
    \sign(\Unlock_a(z))$ satisfies the claim.
\end{proof}

We have therefore established that for any key $a \in [0, B)^n$,  any sample $(\bm z,\bm b)\sim \Dyes^a$,
  with $\bm z = \Lock_{ a}(\bm
x, \bm b)$, satisfies  $\bm b = h_{a}(\bm z)$ and thus, is generated by the halfspace $h_{a}$, 
as desired.

\subsubsection*{The Distribution $\Dno^a$}

It remains to construct the NO distribution $\Dno^a$. This distribution is nearly identical, except
that we choose independently and uniformly random $\bm x \sim \pmset^n$, $\bm b \sim \pmset$, and $\bm b' \sim
\pmset$, and then take the labeled sample
\[
(\bm z,\bm a) = \big( \Lock_a(\bm x, \bm b'), \bm b\big) .
\]
That is, we replace the label $\bm b$ with an independent $\bm b'$ in the locking mechanism, so that
the label $\bm b$ is independent of the locked value; in particular, each sample receives an
independently random label $\bm b$, as desired.

Observe that, in both $\Dyes^a$ and $\Dno^a$, the marginal distribution of $\bm z$ is highly involved, which is how the 
lower bound construction takes advantage of
the distribution-free feature of the model.

\subsection{Proof of Indistinguishability (\cref{lemma:indistinguishability})}
\label{section:indistinguishability}

\subsubsection{Proof Overview}
Recall that the random variable $\TupleS_\yes$ is $m$ samples drawn from $\Dyes^{\bm a}$, where $\ba
\sim [0,B)^n$. Similarly, $\TupleS_\no$ is $m$ samples drawn from $\Dno^{\bm a}$. We will prove
indistinguishability of these two distributions by giving a coupling, \ie a joint distribution $\pi$
of $(\TupleS_\yes, \TupleS_\no)$ such that
\[
    \Pru{(\TupleS_\yes, \TupleS_\no) \sim \pi}{ \TupleS_\yes \neq \TupleS_\no } = o_n(1) .
\]
To construct the coupling, we first let $\bm X \sim \pmset^{m \times n}$ be a random matrix where
each row $\bm x_i \sim \pmset^n$ is a uniformly random hypercube vector, and let $\vec{\bm b}$,
$\vec{\bm b'} \sim \pmset^m$ be independent uniform vectors. The pairs $(\bm x_i,\vec{\bm b}_i)$ and $(\bm x_i,\vec{\bm b}_i')$ will be the inputs to the locking
mechanism.

Let $\bm a \sim [0, B)^n$. In our coupling, we generate $\TupleS_\yes$ by taking the $i$-th labeled
sample as
\[
    \big(\Lock_{\bm a}(\bm x_i, \vec{\bm b}_i), \vec{\bm b}_i\big) = \left( \bm x_i \circ \big\langle \lfloor \bm a^\top \bm x_i\rfloor 
        + \vec{\bm b}_i  \big\rangle, \vec{\bm b}_i \right) .
\]
We now need a way to generate $\TupleS_\no$, with inputs $\vec{\bm b'}$ instead of $\vec{\bm b}$ in
the locking mechanism, such that the labeled samples are identical with high probability. If $\bm X$ is
full rank (which occurs with very high probability when $m = o(n)$), we can solve the linear system
$$
\vec{\bm b} - \vec{\bm b'} = \bm X \bm \Delta
$$
in variables $\bm \Delta\in \bR^n$ to obtain a vector $\bm \Delta\in \bR^n$ that satisfies
\[
    \bm X \ba + \vec{\bm b}  = \bm X (\bm a + \bm \Delta) + \vec{\bm b'}.
\]
Therefore, if we use the vector $\bm a' =
\bm a + \bm \Delta$ as the ``key'' to generate samples in $\TupleS_\no$, we would have
\[
    \big(\Lock_{\bm a + \bm \Delta}(\bm x_i, \vec{\bm b}_i'), \vec{\bm b}_i\big)
    = \big(\Lock_{\bm a}(\bm x_i, \vec{\bm b}_i), \vec{\bm b}_i\big) ,
\]
and the samples would match. However, we also need $\bm a'$ to be uniformly distributed in
$[0,B)^n$, which is not the case. We can achieve this by setting each coordinate $j$ as
$\bm a'_j = \bm a_j + \bm \Delta_j \mod B$, which gives the desired results as
long as none of the coordinates ``wrap around''; the proof concludes by showing that this happens
with low probability, which depends on the singular values of $\bm X$.

\subsubsection{Proof}

We now formalize the above proof overview. Let $m=o(n)$. We define the joint distribution $\pi$ over  pairs of $m$-tuples
$(\TupleS_\yes, \TupleS_\no)$ as follows:
\begin{flushleft}\begin{enumerate}
    \item Sample uniformly and independently $\bx_i \sim \pmset^n$ for $i \in [m]$ and $\vec{\bb}, \vec{\bb}' \sim
        \pmset^m$.
    \item Let $\bX \in \pmset^{m \times n}$ be the matrix whose rows are the vectors $\bx_i$. If the
        smallest singular value $\sigma_{\min}(\bX) \ge c \sqrt{n}$, where $c > 0$ is a sufficiently
        small absolute constant to be chosen later (as we will see this is the typical case), let\footnote{The choice of $\bdelta$ comes from
        applying the pseudoinverse of $\bX$.}
        \begin{equation}
            \label{eq:delta}
            \bdelta \define \bX^\top (\bX \bX^\top)^{-1} (\vec{\bb} - \vec{\bb}') \,,
        \end{equation}
        which satisfies $\bX \bdelta = \vec{\bb} - \vec{\bb}'$. Otherwise, let $\bdelta = \bot$.
    \item In the (typical) case that $\bdelta \neq \bot$, 
        sample $\ba_\yes \sim [0, B)^n$ and let $\ba_\no \define \ba_\yes + \bdelta \mod B$. (The
        mod operation is taken element-wise.)
        Otherwise (if $\bdelta = \bot$), sample independent $\ba_\yes, \ba_\no \sim [0, B)^n$. 
    \item For each $i \in [m]$, let $$\bz_i \define \Lock_{\ba_\yes}(\bx_i, \vec{\bb}_i)\quad\text{and}\quad \bz'_i
        \define \Lock_{\ba_\no}(\bx_i, \vec{\bb}'_i).$$ Output $(\TupleS_\yes, \TupleS_\no)$ where
        $
            \TupleS_\yes  \define \big(({\bz_i, \vec{\bb}_i}):i\in [m]\big)$ and $
 \TupleS_\no \define \big(({\bz'_i, \vec{\bb}_i}):i\in [m]\big)$.
\end{enumerate}\end{flushleft}

The following proposition shows that $\pi$ is indeed a coupling.

\begin{proposition}
    \label{lemma:coupling}
    The joint distribution $\pi$ is a coupling between $\TupleS_\yes$ and $\TupleS_\no$.
\end{proposition}
\begin{proof}
It follows from the construction that the first marginal is
   $\TupleS_\yes$.
For the second,
    conditional on each fixed $\bX$, $\vec{\bb}$, and $\vec{\bb}'$, there are two cases:
    either $\bdelta = \bot$ and $\ba_\no$ is uniformly distributed by definition, or $\bdelta$ is a
    fixed vector and then $\ba_\no = \ba_\yes + \bdelta \mod B$ is again uniformly distributed in
    $[0, B^n)$, since $\ba_\yes$ is and by translation invariance. Thus the second marginal of
    $\pi$ is  $\TupleS_\no$.
\end{proof}

We now complete the proof of \cref{lemma:indistinguishability} by showing that,
when $m = o(n)$, the coupling $\pi$ satisfies
\[
        \Pru{\pbra{\TupleS_\yes, \TupleS_\no} \sim \pi}{\TupleS_\yes \ne \TupleS_\no}
        = o_n(1) \,.
\]
We will use a standard concentration bound on the spectrum of subgaussian matrices:

\begin{theorem}[Special case of \cite{Ver26}, Theorem~4.6.1]
\label{thm:singular-value}
Let $\bm X$ be an $m \times n$ random matrix with independent, mean-zero, subgaussian, isotropic
rows $\bm x_i$. Then for all choices of $c > 0$, the smallest singular value $\sigma_{\min}(\bX)$
    satisfies
\[
    \sigma_{\min}(\bX) \geq \sqrt n - C(\sqrt m + c\sqrt n)
\]
with probability at least $1-2 e^{-c^2n}$, where $C > 0$ is a universal constant.
\end{theorem}

    Given that $C$ is a universal constant, by setting $c$ small enough, we ensure that $\sigma_{\min}(\bm X)\ge c\sqrt{n}$ with probability $1-o_n(1)$. 
    We finish the proof of \cref{lemma:indistinguishability}
    by showing
\begin{flushleft}\begin{enumerate}
\item When $\sigma_{\min}(\bX)\ge c\sqrt{n}$ and $\ba_\no=\ba_\yes+\bm \Delta$, we have $\TupleS_\yes=\TupleS_\no$;\footnote{Note that here, unlike in the construction of $\pi$, 
    ``$\ba_\yes+\bdelta$'' refers to vector addition  without a $\mathrm{mod}\ B$ operation.} and
\item Conditional on $\sigma_{\min}(\bX)\ge c\sqrt{n}$, we have $\ba_\no=\ba_\yes+\bm\Delta$ with probability at least $1-o_n(1)$.
\end{enumerate}\end{flushleft}
\cref{lemma:indistinguishability} then follows by a union bound with the event $\sigma_{\min}(\bX)\ge c\sqrt{n}$.

For the first item, assuming that 
  $\sigma_{\min}(\bX)\ge c\sqrt{n}$ and ${\bm a}_\no={\bm a}_\yes+\bm \Delta$,
    we write
    \[
        \btheta_\yes \define \floor{\bX \ba_\yes} + \vec{\bb}
        \quad \text{and} \quad
        \btheta_\no \define \floor{\bX \ba_\no} + \vec{\bb}' \,.
    \]
    Then the vectors $\bz_i$ (resp.\ $\bz'_i$)
    depend only on the vectors $\bx_i$ and the rows of $\btheta_\yes$ (resp.\ $\btheta_\no$) via
    concatenation and binary representation operations in the locking mechanism
    \eqref{eq:locking-mechanism}.
    Hence it suffices to show
    that $\btheta_\yes = \btheta_\no$. Using the integrality of $\bX \bdelta = \vec{\bb} -
    \vec{\bb}'$, we have
    \[
        \btheta_\no
        = \floor{\bX \ba_\no} + \vec{\bb}'
        = \floor{\bX (\ba_\yes + \bdelta)} + \vec{\bb}'
        = \floor{\bX \ba_\yes} + \bX \bdelta + \vec{\bb}'
        = \floor{\bX \ba_\yes} + \vec{\bb}
        = \btheta_\yes \,,
    \]
    as claimed.

    Next, we prove the second item. It suffies to show that $\ba_\yes + \bdelta \in [0, B)^n$ holds with probability $1-o_n(1)$, conditioning on $\sigma_{\min}(\bX)\ge c\sqrt{n}$.
    By \eqref{eq:delta} we have, using the
    Cauchy-Schwarz inequality,
    \[
        \vabs{\bdelta}_2
        = \vabs{\bX^\top (\bX \bX^\top)^{-1} (\vec{\bb} - \vec{\bb}')}_2
        \le \vabs{\bX^\top (\bX \bX^\top)^{-1}}_2 \vabs{\vec{\bb} - \vec{\bb}'}_2
        \le 2\sqrt{m} \cdot \vabs{\bX^\top (\bX \bX^\top)^{-1}}_2 \,.
    \]
    To bound the final term, write the singular-value decomposition $\bX = U \Sigma V^\top$,
    where $U \in \bR^{m \times m}$, $\Sigma \in \bR^{m \times n}$ and $V \in \bR^{n \times n}$, and
    note that
    \[
        \bX^\top (\bX \bX^\top)^{-1}
        = V \Sigma^\top U^\top (U \Sigma \Sigma^\top U^\top)^{-1}
        = V \Sigma^\dagger U^\top \,,
    \]
    where $\Sigma^\dagger \in \bR^{n \times m}$ is diagonal with $\Sigma^\dagger_{i,i} = 1 /
    \Sigma_{i,i}$. Thus $\vabs{\bX^\top (\bX \bX^\top)^{-1}}_2 = \frac{1}{\sigma_{\min}(\bX)}$ and by \Cref{thm:singular-value},
    \[
        \vabs{\bdelta}_\infty
        \le \vabs{\bdelta}_2
        \le \frac{2\sqrt{m}}{c\sqrt{n}}
        \le \frac{1}{c} \,.
    \]
    Therefore a sufficient condition for $\ba_\yes + \bdelta \in [0, B)^n$ is that $\ba_\yes \in
    (c^{-1}, B - c^{-1})^n$, and by a union bound over the $n$ independent coordinates of $\ba_\yes$ this holds with probability at least
    \[
        1 - n \cdot \frac{2 c^{-1}}{B} = 1 - \frac{2}{cn^2}=1-o_n(1)
    \]
    by our choice of $B = n^3$. 
    This finishes the proof of the second item and that of 
    \cref{lemma:indistinguishability}.

\section{Upper Bound Proof}
\label{section:ub}

\newcommand{\COMP}{\mathsf{comp}}
\newcommand{\DECOMP}{\mathsf{decomp}}

We prove a general upper bound on one-sided-error distribution-free sample-based testers, using
\emph{stable sample-compression schemes}.
Sample-compression schemes are attributed to \cite{LW86}. \emph{Stable}
sample-compression is a strengthening defined by \cite{BHMZ20}:

\begin{definition}[Stable Sample-Compression Scheme]\label{def:hehe2}
Let $\cH$ be a class of functions $\cX \to \pmset$. A \emph{stable sample-compression scheme} of size
$\ell$ for $\cH$ is a pair of algorithms $(\COMP, \DECOMP)$ such that, for all $m$ and every labeled
set of samples $S = \{ (x_1, h(x_1)), \dotsc, (x_m, h(x_m)) \}$ with $h \in \cH$:
\begin{enumerate}
\item $\COMP(S)$ outputs a set $\COMP(S) \subseteq S$ with $|\COMP(S)| \leq \ell$.
\item $\DECOMP(\COMP(S))$ outputs a function $g \in \cH$ such that $g(x_i) = h(x_i)$ on all $(x_i,
    h(x_i)) \in S$.
\item For any $S'$ with $\COMP(S) \subseteq S' \subseteq S$, we have $\COMP(S) = \COMP(S')$.
\end{enumerate}
\end{definition}

We must also define certificates of non-membership for a fixed class of functions $\cH$. For a set
$C \subseteq \cX$ and function $f \colon \cX \to \pmset$, we write $f(C) \define \{ (x, f(x)) \mid x
\in C \}$ for the $f$-labeled set of samples corresponding to $C$.

\begin{definition}[Certificates and minimal certificates]
Let $\cH$ be any set of functions $\cX \to \pmset$ and let $f \colon \cX \to \pmset$ satisfy $f
\notin \cH$. A set $C \subseteq \cX$ is a \emph{certificate of non-membership} for $f$ if $f(C) \neq
h(C)$ for all $h \in \cH$. A certificate $C$ is \emph{minimal} if none of its proper subsets are
certificates. The \emph{minimal certificate size} of $f$ is the size of the largest \emph{finite} minimal
certificate; for convenience, we say the minimal certificate size is 0 if all certificates are infinite.
\end{definition}

Then our general theorem is the following:

\begin{boxtheorem}
\label{thm:general-1-sided}
Let $\cH$ be any class of functions $\cX \to \pmset$ which admits a stable sample-compression scheme
of size $\ell$, and where every $f \notin \cH$ has minimal certificate size at most $c$. Then
there is a distribution-free sample-based tester for $\cH$ with one-sided error and sample
complexity $O\left(\frac{\ell + c + \log(1/\delta)}{\epsilon}\right)$.
\end{boxtheorem}

\begin{remark}
\label{rem:finitely-certifiable}
We have carefully defined the minimal certificate size of $f \notin \cH$ to
count only the \emph{finite} minimal certificates. A function may have
certificates of infinite size for which no subset is minimal, even if it also
has finite certificates. For example, the function $f \colon \bR^n \to \pmset$
which takes value 1 on exactly the set $\{ x \mid (x_2 > 0 \wedge x_1 > 0) \vee
(x_2 \leq 0 \wedge x_1 \geq 0 ) \}$ is not a halfspace, but it appears not to
have any minimal certificates, and it is consistent with a halfspace on every
finite set $X \subset \bR^n$. We say that such a function has minimal
certificate size 0. If we swap the example function's values on the points
$-\vec 1, \vec 0$, it has certificates of infinite size, and also a minimal
certificate $(-\vec 1, \vec 0, \vec 1)$ of size 3, so its minimal certificate size is 3.
\end{remark}

\subsection{Upper Bound for Halfspaces}

\cref{thm:general-1-sided} implies our upper bound for testing halfspaces
with one-sided error.

\begin{corollary}
\label{cor:halfspaces}
For any $\epsilon, \delta > 0$,
there is a distribution-free sample-based halfspace tester over $\bR^n$ with one-sided error, failure probability
at most $\delta$ and sample complexity $O\left(\frac{n + \log(1/\delta)}{\epsilon}\right)$.
\end{corollary}
\begin{proof}
As stated in \cite{BHMZ20}, the SVM algorithm implies a stable sample-compression
scheme of size $n+1$ for halfspaces; the proof is attributed to \cite{VC74,LL20}. So,
to apply \cref{thm:general-1-sided}, we must verify that for every $f \notin \cH$, its finite minimal certificates have size at most $n+2$. This is exactly Kirchberger's theorem
(see \eg \cite{RS50}), which states the following.  Let $X \subseteq \bR^n$ be any finite set and $f
\colon X \to \pmset$ any function that does not agree with any halfspace. Then there exists $C
\subseteq X$ of size $|C| \leq n+2$ on which $f$ does not agree with any halfspace.
\end{proof}

\subsection{Proof of General Upper Bound}

We require the following theorem of \cite{BHMZ20}, which is not stated
explicitly in the following form but can be seen by examining the algorithm in
their proof. Recall that a \emph{proper} learning algorithm is
one which (with probability at least $1-\delta$) outputs a function $h \in \cH$
such that $\dist_\cD(f,h) < \epsilon$.

\begin{theorem}[\cite{BHMZ20}, Theorem 15]
\label{thm:bhmz}
Let $\cH$ be any class of functions with a stable sample-compression scheme $(\COMP, \DECOMP)$ of
size $\ell$. Then the algorithm which, given labeled sample $S$, outputs $\DECOMP(\COMP(S))$, is a
proper learner for $\cH$ with sample complexity $O\left(\frac{\ell +
\log(1/\delta)}{\epsilon}\right)$.
\end{theorem}

To prove \cref{thm:general-1-sided} we establish the following:

\begin{proposition}
\label{claim:compression}
Let $\cH$ be any set of functions $\cX \to \pmset$ which admits a stable sample-compression scheme
of size $\ell$, and let $f \colon \cX \to \pmset$ be any function with minimal certificate
size $c$. Then the class $\cH \cup \{f\}$ admits a stable sample-compression scheme $(\COMP',
\DECOMP')$ of size $\max\{\ell,c\}$, with the property that $
    \DECOMP'(\COMP'(S)) = f
$
if and only if $S$ contains a certificate of non-membership for $f$.
\end{proposition}
\begin{proof}
Let $(\COMP, \DECOMP)$ be the stable sample compression scheme for $\cH$.  For any fixed $f \notin
\cH$, we design the following stable sample-compression scheme $(\COMP', \DECOMP')$ for $\cH \cup \{f\}$. Let $\cC$ be
the set of finite minimal certificates for $f$. We choose an arbitrary total order on $\cC$ (this can be done without the axiom of choice whenever $\cX$ itself has a total order; since the minimal certificates are finite subsets of $\cX$, they inherit a lexicographic ordering from $\cX$). 
The compression algorithm $\COMP$ is then defined as follows.  Given a set of labeled samples $S = \{(x_1, g(x_1)),
\dotsc, (x_m, g(x_m))\}$:
\begin{flushleft}\begin{enumerate}
    \item If there exists $h \in \cH$ which agrees with the labels $g(x_i)$ on all $x_i$,
        output $\COMP(S)$ (the original sample compression for $\cH$);
    \item Otherwise, if the labels do not agree with any $h \in \cH$, they must agree with $f$, so
        there must exist a (finite) minimal certificate $C$ for $f$ within the sample. Let $C'$ be the
        minimal certificate which is also minimal in the total ordering on $\cC$ among all certificates contained in the sample, and output
        $\COMP'(S) = C'$.
\end{enumerate}\end{flushleft}
The decompression algorithm $\DECOMP$ is as follows: Given a set of labeled samples $T = \{(x_1, g(x_1)),$ $ \dotsc,
(x_t, g(x_t))\}$ (with $t \leq \max(\ell, c)$):
\begin{enumerate}
\item If there exists $h \in \cH$ such that $h(x_i) = g(x_i)$ for all $x_i$, output $\DECOMP(T)$.
\item Otherwise, if the labels do not agree with any $h \in \cH$, they must agree with $f$, so
output $f$.
\end{enumerate}
We must now show correctness of this compression scheme. It is clear that $|\COMP'(S)| \leq
\max\{\ell,c\}$ so the first item of \Cref{def:hehe2} follows.
For the second and third items, 
let $S$ be any set of labeled samples. We first consider the case that $S$ is consistent with some $h \in \cH$. In
this case, we have $\COMP'(S) = \COMP(S)$ and
\[
    \DECOMP'(\COMP'(S)) = \DECOMP'(\COMP(S)) = \DECOMP(\COMP(S)) ,
\]
so the second and third items of \Cref{def:hehe2} follow from 
those of the original scheme $\COMP$.

Now suppose $S$ is inconsistent with all $h \in \cH$, and therefore consistent with $f$. Then
$\COMP'(S)$ is a minimal certificate $C$ for $f$ which is also the smallest according to the total
order on $\cC$, and we have  $\DECOMP'(\COMP'(S)) = \DECOMP'(C) = f$, from which the second item follows. On the other hand, for any $S'$
which satisfies $\COMP'(S) = C \subseteq S' \subseteq S$, $C$ must again be the smallest minimal
certificate in $S'$ according to the total order on $\cC$, so $\COMP'(S') = \COMP'(S)$, as desired.
\end{proof}

With the proper learning theorem of \cite{BHMZ20} and the above claim, we can complete the proof of
\cref{thm:general-1-sided}. Let $\cH$ be any class of functions which admits a stable
sample-compression scheme of size $\ell$, and for which every finite minimal certificate for any $f
\notin \cH$ has size at most $c$. Let $m = \Theta\left(\frac{\ell + c +
\log(1/\delta)}{\epsilon}\right)$, $\cD$ be any distribution over $\cX$, and $f
\colon \cX \to \pmset$ be any function that is $\epsilon$-far from $\cH$ under $\cD$. We must show
that, with probability at least $1-\delta$, a random set of labeled samples $\bm S = \{(\bm x_1,
f(\bm x_1)), \dotsc, (\bm x_m, f(\bm x_m))\}$ contains a certificate of non-membership for $f$.

By \cref{claim:compression}, there exists a stable sample-compression scheme $(\COMP,
\DECOMP)$ for $\cH \cup \{f\}$ with size $\max\{\ell,c\}\le \ell + c$, such that $\DECOMP(\COMP(S)) = f$ if
and only if $S$ contains a certificate of non-membership for $f$. By \cref{thm:bhmz}, the algorithm
$\DECOMP(\COMP(S))$ is a proper learner for $\cH \cup \{f\}$.

Let $\bm S$ be the random samples labeled by $f$. Then with probability at least $1-\delta$,
we have the event
\begin{equation}
\label{eq:distance}
    \dist_\cD\left(f, \DECOMP\big(\COMP(\bm S)\big)\right) < \epsilon ,
\end{equation}
since $\DECOMP(\COMP(\cdot))$ is a proper learning algorithm. Since $f$ is $\epsilon$-far from $\cH$
under $\cD$ and the learning algorithm is proper, it must be the case that, if event
\eqref{eq:distance} occurs, then $\DECOMP(\COMP(\bm S)) = f$.
By
\cref{claim:compression},
this occurs if and only if $\bm S$ contains a certificate of non-membership for $f$.

\subsection*{Disclosure of AI Use}

The lower bound proof in \cref{section:main-lb} for domain $\pmset^n$ was AI assisted; see
\cref{remark:ai} for details.  All text in this paper was handcrafted with love and care by humans,
by the ancient traditional method of pushing buttons on keyboards and then arguing about it.

\bibliographystyle{alpha}
\bibliography{references}

\providecommand{\FOCS}{Proceedings of the IEEE Symposium on Foundations of
  Computer Science (FOCS)} \providecommand{\SODA}{Proceedings of the ACM-SIAM
  Symposium on Discrete Algorithms (SODA)} \providecommand{\STOC}{Proceedings
  of the ACM SIGACT Symposium on Theory of Computing (STOC)}
  \providecommand{\ITCS}{Proceedings of the Innovations in Theoretical Computer
  Science Conference (ITCS)} \providecommand{\ICALP}{Proceedings of the
  International Colloquium on Automata, Languages, and Programming (ICALP)}
  \providecommand{\ICML}{Proceedings of the International Conference on Machine
  Learning (ICML)} \providecommand{\COLT}{Proceedings of the Conference on
  Learning Theory (COLT)} \providecommand{\AISTATS}{Proceedings of the
  International Conference on Artificial Intelligence and Statistics (AISTATS)}
  \providecommand{\TOCT}{ACM Transactions on Computation Theory (TOCT)}
  \providecommand{\RANDOM}{Approximation, Randomization, and Combinatorial
  Optimization. Algorithms and Techniques (APPROX/RANDOM)}
  \providecommand{\JACM}{Journal of the ACM (JACM)}
  \providecommand{\SIAMJOC}{SIAM Journal on Computing}
  \providecommand{\TOC}{Theory of Computing} \providecommand{\TOIT}{IEEE
  Transactions on Information Theory} \providecommand{\COLT}{Proceedings of the
  Conference on Learning Theory (COLT)} \providecommand{\ALT}{Proceedings of
  Algorithmic Learning Theory (ALT)} \providecommand{\NEURIPS}{Advances in
  Neural Information Processing Systems (NeurIPS)}
  \providecommand{\JMLR}{Journal of Machine Learning Research}
  \providecommand{\SOSA}{Symposium on Simplicity in Algorithms (SOSA)}
  \providecommand{\SICOMP}{{SIAM} Journal on Computing (SICOMP)}
\begin{thebibliography}{AAHLZ24}

\bibitem[AAHLZ24]{AHMLZ24}
Ishaq Aden-Ali, Mikael~M{\o}ller H{\o}gsgaard, Kasper~Green Larsen, and Nikita
  Zhivotovskiy.
\newblock Majority-of-three: The simplest optimal learner?
\newblock In {\em \COLT}, 2024.

\bibitem[Bal97]{ball1997elementary}
Keith Ball.
\newblock An elementary introduction to modern convex geometry.
\newblock {\em Flavors of geometry}, 31(1--58):26, 1997.

\bibitem[BBBY12]{BBBY12}
Maria{-}Florina Balcan, Eric Blais, Avrim Blum, and Liu Yang.
\newblock Active property testing.
\newblock In {\em \FOCS}, 2012.

\bibitem[BFH21]{BFH21}
Eric Blais, Renato {Ferreira Pinto Jr.}, and Nathaniel Harms.
\newblock {VC} dimension and distribution-free sample-based testing.
\newblock In {\em \STOC}, 2021.

\bibitem[BHK26]{BHK26}
Lorenzo Beretta, Nathaniel Harms, and Caleb Koch.
\newblock Feature selection and junta testing are statistically equivalent.
\newblock In {\em \SODA}, 2026.

\bibitem[BHMZ20]{BHMZ20}
Olivier Bousquet, Steve Hanneke, Shay Moran, and Nikita Zhivotovskiy.
\newblock Proper learning, {Helly} number, and an optimal {SVM} bound.
\newblock In {\em \COLT}, 2020.

\bibitem[Can22]{Can22}
Cl{\'e}ment Canonne.
\newblock {\em Topics and techniques in distribution testing}.
\newblock now Publishers, 2022.

\bibitem[CP22]{CP22}
Xi~Chen and Shyamal Patel.
\newblock Distribution-free testing for halfspaces (almost) requires {PAC}
  learning.
\newblock In {\em \SODA}, 2022.

\bibitem[ES20]{ES20}
Rogers Epstein and Sandeep Silwal.
\newblock Property testing of lp-type problems.
\newblock In {\em \ICALP}, 2020.

\bibitem[FH26]{FH26}
Renato {Ferreira Pinto Jr.} and Nathaniel Harms.
\newblock Testing support size more efficiently than learning histograms.
\newblock {\em TheoretiCS}, 5:10, 2026.

\bibitem[GGR98]{GGR98}
Oded Goldreich, Shafi Goldwasser, and Dana Ron.
\newblock Property testing and its connection to learning and approximation.
\newblock {\em \JACM}, 45(4):653--750, 1998.

\bibitem[GR16]{GR16}
Oded Goldreich and Dana Ron.
\newblock On sample-based testers.
\newblock {\em \TOCT}, 8(2):7:1--7:54, 2016.

\bibitem[Han16]{Han16}
Steve Hanneke.
\newblock The optimal sample complexity of {PAC} learning.
\newblock {\em Journal of Machine Learning Research}, 17(38):1--15, 2016.

\bibitem[Han19]{Han19}
Steve Hanneke.
\newblock {CSTheory} stackexchange answer.
\newblock
  \url{https://cstheory.stackexchange.com/questions/40161/proper-pac-learning-vc-dimension-bounds},
  2019.
\newblock Accessed 2026-07-28.

\bibitem[Har19]{Har19}
Nathaniel Harms.
\newblock Testing halfspaces over rotation-invariant distributions.
\newblock In {\em \SODA}, 2019.

\bibitem[LL20]{LL20}
Philip~M Long and Raphael~J Long.
\newblock On the complexity of proper distribution-free learning of linear
  classifiers.
\newblock In {\em \ALT}, 2020.

\bibitem[LW86]{LW86}
Nick Littlestone and Manfred Warmuth.
\newblock Relating data compression and learnability.
\newblock Unpublished manuscript, 1986.

\bibitem[MORS10]{MORS10}
Kevin Matulef, Ryan O'Donnell, Ronitt Rubinfeld, and Rocco~A. Servedio.
\newblock Testing halfspaces.
\newblock {\em \SICOMP}, 39(5):2004--2047, 2010.

\bibitem[RRSS09]{RRSS09}
Sofya Raskhodnikova, Dana Ron, Amir Shpilka, and Adam Smith.
\newblock Strong lower bounds for approximating distribution support size and
  the distinct elements problem.
\newblock {\em \SICOMP}, 39(3):813--842, 2009.

\bibitem[RS50]{RS50}
Hans Rademacher and IJ~Schoenberg.
\newblock Helly's theorems on convex domains and {Tchebycheff}'s approximation
  problem.
\newblock {\em Canadian Journal of Mathematics}, 2:245--256, 1950.

\bibitem[VC71]{VC71}
V.~N. Vapnik and A.~Ya. Chervonenkis.
\newblock On the uniform convergence of relative frequencies of events to their
  probabilities.
\newblock {\em Theory of Probability \& Its Applications}, 16(2):264--280,
  1971.

\bibitem[VC74]{VC74}
V.~N. Vapnik and A.~Ya. Chervonenkis.
\newblock {\em Theory of pattern recognition}.
\newblock Nauka, Moscow, 1974.

\bibitem[Ver26]{Ver26}
Roman Vershynin.
\newblock {\em High-Dimensional Probability: An Introduction with Applications
  in Data Science}.
\newblock Cambridge Series in Statistical and Probabilistic Mathematics.
  Cambridge University Press, 2 edition, 2026.

\bibitem[VV11]{VV11}
Gregory Valiant and Paul Valiant.
\newblock Estimating the unseen: an $n/log(n)$-sample estimator for entropy and
  support size, shown optimal via new {CLTs}.
\newblock In {\em \STOC}, 2011.

\bibitem[WY19]{WY19}
Yihong Wu and Pengkun Yang.
\newblock Chebyshev polynomials, moment matching, and optimal estimation of the
  unseen.
\newblock {\em The Annals of Statistics}, 47(2):857--883, 2019.

\end{thebibliography}

\appendix

\section{Alternate Lower Bound Proof}
\label{section:alternate}

We give an alternate, more intuitive proof of the main lower bound, but with domain $\bR^n$ instead
of $\pmset^n$ (so the result is weaker).

\begin{boxtheorem}
\label{thm:main-reals}
Distribution-free sample-based halfspace testing over $\bR^n$ requires $\Omega \left ( \frac{n + \log(1/\delta)}{\epsilon} \right)$ samples.
\end{boxtheorem}

\subsection{Reduction from Testing Halfspace Truncation}

We transform the halfspace testing problem into a distribution testing
problem, which lets us ignore the labels of the points. For any (absolutely continuous) probability
distribution $\cD$ over $\bR^n$ with probability density $\mu$, let $\supp(\cD) = \{
x \colon \mu(x) > 0 \}$ denote its support.

\begin{definition}[Halfspace Truncation Testing Problem]
Given samples from an unknown probability distribution $\cD$ over $\bR^n$, the output of the
algorithm should satisfy the following with probability at least $3/4$:
\begin{itemize}
\item If $\exists w \in \bR^n$ such that
$\supp(\cD) \subseteq \{ x \colon \inn{w,x} \geq 0 \}$, the algorithm outputs $\Accept$;
\item If $\forall w \in \bR^n$, $\cD$ has total probability mass at least
$\epsilon$ on the set $\{ x \colon \inn{w,x} \leq 0 \}$, the algorithm outputs $\Reject$.
\end{itemize}
\end{definition}

We will prove the following theorem.
\begin{boxtheorem}
\label{thm:halfspace-truncation}
Testing halfspace truncation in $\bR^n$ requires $\Omega(n/\epsilon)$ samples.
\end{boxtheorem}

This implies \cref{thm:main-reals} in the case where $\delta = \frac{1}{3}$ via the following reduction (modulo some technicalities, see
\cref{rem:redux}).

\begin{proposition}[Reduction]
\label{prop:truncation-to-halfspace}
If there is a distribution-free sample-based halfspace tester in $\bR^n$ with sample complexity $m
= m(n,\epsilon)$, then there is a tester for halfspace truncation with sample complexity
$m(n,\epsilon/2)$.
\end{proposition}
\begin{proof}
We perform a reduction as follows. Given a sample $\bm x \sim \cD$ for the halfspace truncation
problem, provide the sample $(\bm b \bm x, \bm b)$ to the halfspace tester, where $\bm b \sim
\pmset$ is an independent uniform random sign. Let $\cD'$ be the distribution of $\bm b \bm x$
defined in this way.  To prove correctness of the reduction, first assume $\cD$ is supported on a
halfspace through the origin, witnessed by normal vector $w \in \bR^n$. Then each sample $(\bm b \bm
x, \bm b)$ given to the halfspace tester is distributed as $( \bm b \bm x, \sign\inn{\bm b \bm x,
w}) \equiv (\bm x', \sign \inn{\bm x', w})$ where $\bm x' \sim \cD'$. So the tester should output
$\Accept$.

We will now argue by contrapositive that the distribution-free halfspace tester must reject if $\cD$ is far from being supported on a halfspace. Indeed, suppose that there is a halfspace $h(x) = \sign(\inn{w,x}-t)$ such that $\Pr{ h(\bm b \bm
x) \neq \bm b } < \epsilon/2$ (\ie we are in a case where the halfspace tester is allowed to output
$\Accept$; but see \cref{rem:redux}). Then, since $\cD'$ is symmetric (\ie $y$ and $-y$ have the same
probability density) and $h(x)$ is constant on $|\inn{w,x}| \leq t$,
\begin{align*}
    \frac{\epsilon}{2}
    > \Pr{ h(\bm b \bm x) \neq \bm b }
    &= \frac{1}{2} \Pr{ |\inn{w,\bm b\bm x}| \leq t }
        + \Pr{ |\inn{w, \bm b \bm x}| > t \;\wedge\; h(\bm b \bm x) \neq \bm b} \\
    &\geq \frac{1}{2} \Pr{ \bm b \neq \sign\inn{w,\bm b \bm x} \wedge |\inn{w,\bm b\bm x}| \leq t}
        + \frac{1}{2} \Pr{ |\inn{w, \bm b \bm x}| > t \;\wedge\; \sign\inn{w,\bm b\bm x} \neq \bm b} \\
    &= \frac{1}{2} \Pr{ \sign\inn{w, \bm b \bm x} \neq \bm b }
    = \frac{1}{2} \Pr{ \sign\inn{w, \bm x} \neq 1 } ,
\end{align*}
so $\cD$ is $\epsilon$-close to supported on a halfspace through the origin.
\end{proof}

\begin{remark}[Technicalities redux]
    \label{rem:redux}
    As in the earlier lower bound for $\pmset^n$, we have technicalities to take
    care of. In the NO case, our lower bound construction again assigns random
    labels to points $x \in \bR^n$, meaning that labels are not assigned by a
    fixed function $f$. In the continuous domain we can partition $\bR^n$ into
    cells of arbitrarily small measure such that (with arbitrarily high
    probability) no two samples will appear within the same cell; then we can
    assign random labels to each cell, and the result defines a fixed function
    $f$.

    Likewise, the dependence on $\epsilon$ may be reintroduced by putting $1-\epsilon$ mass on a
``dummy point'' (or small ball), so it suffices to show an $\Omega(n)$
    lower bound for constant $\epsilon$. The case of general $\delta$ will follow by combining with the $\Omega \left(\frac{1}{\epsilon} \log(1/\delta) \right)$ lower bound in \Cref{section:technicalities}.
\end{remark}

\subsection{Lower-Bound Construction}

\newcommand{\Slab}{\mathsf{Slab}}
\newcommand{\SLAB}{\mathsf{SLAB}}

To prove our lower bound on testing halfspace truncation, we define the following class of
distributions. For any $n$-dimensional unit vector $w \in \bS^{n-1}$, define
\[
    \Slab_w \define \{ x \in \bR^n \colon 0 < \inn{w,x} < t \}
        ,\qquad\text{ where we choose } t \define \frac{1}{n^{1/4}} .
\]
Define the distribution $\cS_w$ as the distribution of $\bm z \sim \cN(0, I)$ conditional on
$z \in \Slab_{w}$. Let $\SLAB \define \{ \cS_w \colon w \in \bS^{n-1} \}$ be the set of all such
slab distributions. Our goal is to show that $m$ samples from a randomly chosen $\cS_{\bm w}$ are
indistinguishable from $m$ samples from $\cN(0,I)$:

\begin{lemma}[Indistinguishability]
\label{lemma:continuous-slab-indistinguishability}
Any algorithm given samples from an unknown distribution $\cD$, which distinguishes (with
probability at least $3/4$) between the cases $\cD = \cN(0, I)$ and $\cD \in \SLAB$
must use at least $\Omega(n)$ samples.
\end{lemma}

The distribution $\cN(0,I)$ is $1/2$-far from being supported on a halfspace through the origin, so
\cref{lemma:continuous-slab-indistinguishability} combined with \cref{prop:truncation-to-halfspace}
and \cref{rem:redux} gives \cref{thm:main-reals}.
The remainder of this section will prove \cref{lemma:continuous-slab-indistinguishability}.

To prove \cref{lemma:continuous-slab-indistinguishability}, we use the following standard technique
in distribution testing (see \eg the book \cite{Can22}).  Let $p, q$ be the probability density
functions (PDFs) of any two probability distributions (each of which is absolutely continuous with respect to some
reference measure $\mu$). The $\chi^2$ divergence is defined as
\[
    \chi^2(p,q) \define \int_x \frac{(p(x)-q(x))^2}{q(x)} d\mu(x) .
\]
A standard inequality is:
\begin{equation}
\label{eq:tv-to-chi2}
\TV(p,q) \leq \frac{1}{2} \sqrt{\chi^2(p,q)} .
\end{equation}

Let $p_w \define \cS_w^{\otimes m}$ denote the distribution of $m$ samples chosen from $\cS_w$,
where $w$ is a unit vector, and let $p$ denote the distribution of $m$ samples drawn from $p_{\bm
w}$ where the unit vector $\bm w$ is chosen uniformly at random (\ie we draw $\bw$ uniformly at random
and then draw $m$ samples from $\cS_{\bw}$).  Let $q \define \cN(0,I)^{\otimes m}$ be the distribution
of $m$ samples chosen from $\cN(0, I)$. Then, using $X \in (\bR^n)^m$ to denote a sequence of $m$
sample vectors in $\bR^n$, a standard calculation (easy to verify) reveals
\begin{equation}
\label{eq:chi-square-formula}
\chi^2(p,q)
= \Exu{\bm X \sim q}{ \frac{p(\bm X)^2}{q(\bm X)^2} } - 1 .
\end{equation}
Therefore, to prove \cref{lemma:continuous-slab-indistinguishability}, it suffices to show $\Exu{\bm
X \sim q}{ \frac{p(\bm X)^2}{q(\bm X)^2}} = 1 + o(1)$ when $m = o(n)$.

\subsection{Converting to a 2-Dimensional Integral}
The main step of the proof is to convert the $\chi^2$ divergence into a 2-dimensional integral that we can approximate with basic trigonometry.
Since $p_w$ is $q$ conditional on all samples lying within $\Slab_{w}$,
\begin{equation}
\label{eq:p_w}
    \forall X \colon \qquad p_w(X)
        = q(X) \cdot \frac{\ind{X \subset \Slab_{w}}}
                {\Pru{\bm Y \sim q}{ \bm Y \subset \Slab_{w} }}
        = q(X) \cdot \frac{\ind{X \subset \Slab_{w}}} {S^m} ,
\end{equation}
where we define 
\[
    S \define \Pru{\bm g \sim \cN(0,1)}{0 < \bm g < t}.
\]
Therefore, for independent uniformly random unit vectors $\bm w_1, \bm w_2$, we combine
\cref{eq:p_w} with \cref{eq:chi-square-formula} to get
\begin{align}
1 + \chi^2(p,q)
&= \Exu{\bm X \sim q}{ \frac{p(\bm X)^2}{q(\bm X)^2} }
= \Exu{\bm X \sim q}
    { \frac{\Exu{\bm w_1, \bm w_2}{ p_{\bm w_1}(\bm X) p_{\bm w_2}(\bm X) }}{ q(\bm X)^2} } \nonumber \\
&= \frac{1}{S^{2m}} \cdot \Exu{\bm w_1, \bm w_2}{ \Pru{\bm X \sim q}{ \bm X \subseteq \Slab_{\bm
w_1} \cap \Slab_{\bm w_2} }} \nonumber \\
&= \frac{1}{S^{2m}} \cdot \Exu{\bm w_1, \bm w_2}{ \Pru{\bm y \sim \cN(0,I)}{ \bm y \in \Slab_{\bm
w_1} \cap \Slab_{\bm w_2} }^m} . \label{eq:apple}
\end{align}
By rotation-invariance, the expression inside the expectation depends only on the angle $\bm \rho$
between unit vectors $\bm w_1$ and $\bm w_2$. We may therefore assume $\bm w_1$ is the first
standard basis vector and replace the expectation with
\[
    \Exu{\bm \rho}{ P(\bm \rho)^m }
    \define \Exu{\bm \rho}{ \Pru{\bm y \sim \cN(0,I_2)}{ 0 < \bm y_1 < t \;\wedge\; 0 < \bm w^\top \bm y < t }^m } ,
\]
where $\bm \rho$ is the inner product of two uniformly random $n$-dimensional unit vectors, $\bm y
\sim \cN(0, I_2)$ is a 2-dimensional Gaussian and $\bm w = (\bm \rho, \sqrt{1-\bm \rho^2})$.
Observing that $P(|\rho|) \geq P(-|\rho|)$ and the distribution of $\bm \rho$ is symmetric around 0, we have
\begin{equation}
    \Exu{\bm \rho}{P(\bm \rho)^m} \leq 
    \Ex{ P(|\bm \rho|)^m } \label{eq:banana}
\end{equation}
and it suffices to consider the values $\rho > 0$. For fixed $\rho > 0$ and corresponding vector
$w =(\rho,\sqrt{1-\rho^2})$, we define the following sets in $\bR^2$ (see \cref{fig:trig} for an illustration of these
sets):
\begin{align*}
A(\rho)
&\define \left\{ y \in \bR^2 \colon 0 < y_1 < t \;\wedge\; y_2 < 0 \;\wedge\; 0 < w^\top y < t \right\}
    \\
B(\rho)
&\define \left\{ y \in \bR^2 \colon 0 < y_1 < t \;\wedge\; 0 < y_2 < t \;\wedge\; t \leq w^\top y \right\}
    \\
C(\rho)
&\define \left\{ y \in \bR^2 \colon 0 < y_1 < t \;\wedge\; t < y_2 \;\wedge\; 0 < w^\top y < t  \right\} .
\end{align*}
Let $\gamma_2(\cdot)$ denote the 2-dimensional Gaussian probability measure. Then, combining \Cref{eq:apple} and \Cref{eq:banana}, 
\begin{equation}
\label{eq:2d-gaussian-integral}
\begin{aligned}
    1 + \chi^2(p,q)
        &\leq 
        \Exuc{\bm \rho}{ \left(1 + \frac{1}{S^2}(\gamma_2(A(\bm \rho)) +
\gamma_2(C(\bm \rho)) - \gamma_2(B(\bm \rho))) \right)^m }{\bm \rho > 0} \\
        &= 2 \int_0^1 \left(1 + \frac{1}{S^2}(\gamma_2(A(\rho)) + \gamma_2(C(\rho)) - \gamma_2(B(\rho)))\right)^m
            d\sigma_n(\rho) ,
\end{aligned}
\end{equation}
where $\sigma_n$ is the PDF of the first coordinate of a uniformly random $n$-dimensional unit
vector.

\begin{figure}[h!]
\begin{center}
\includegraphics[scale=0.60]{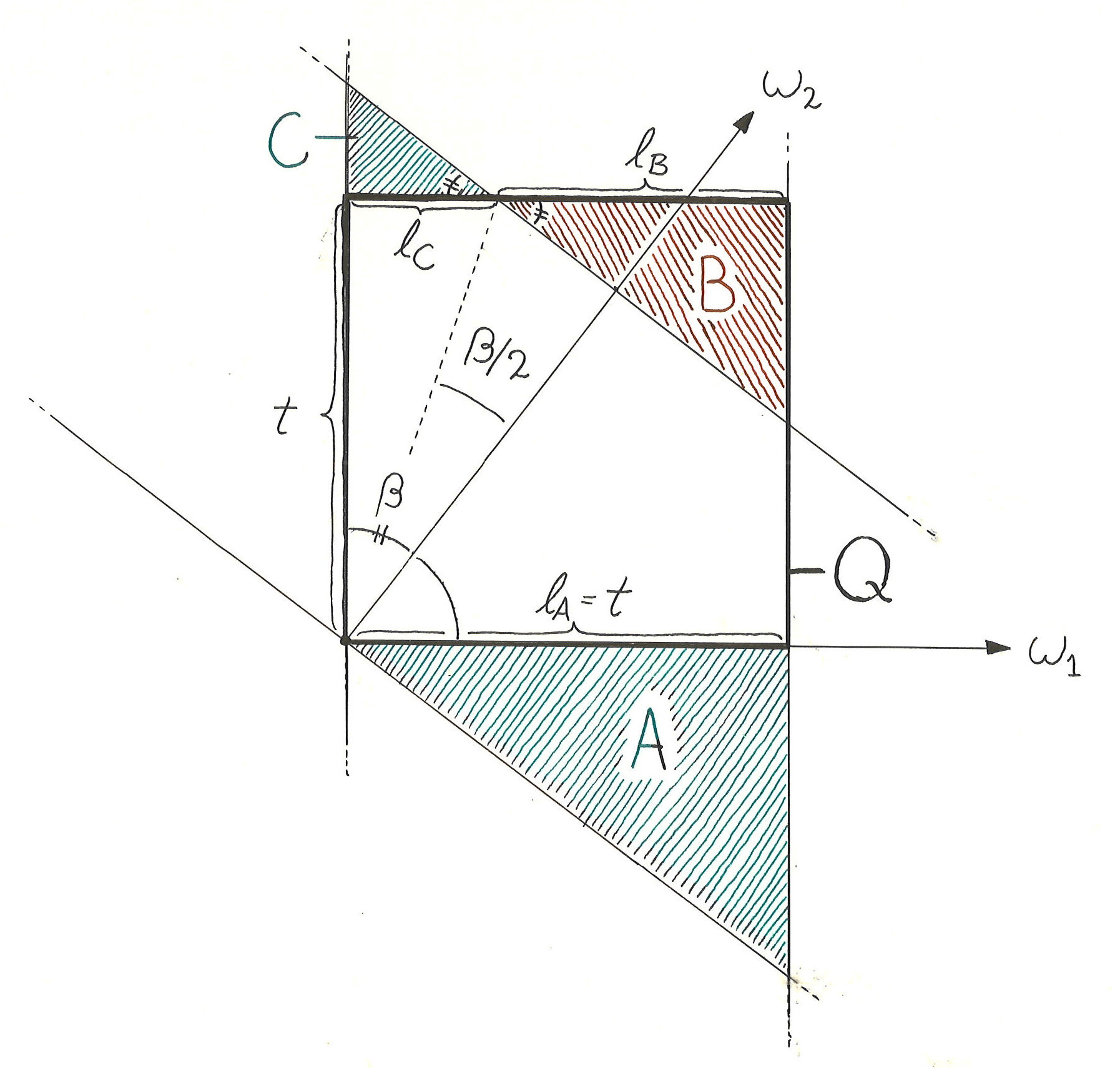}
\end{center}
\caption{We need to bound the 2-dimensional Gaussian integral over the the intersection of two slabs
defined by normal vectors $\bm w_1, \bm w_2$ in \cref{eq:apple}; the intersection is drawn in this figure as the
area $(Q \cup A \cup C) \setminus B$ in the subspace spanned by $\bm w_1, \bm w_2$. The side lengths
$\ell_A=t$, $\ell_B$, $\ell_C=t-\ell_B$ are the sides of the triangles abutting the central square
$Q$, whose Gaussian mass is $\gamma_2(Q)=S^2$.}
\label{fig:trig}
\end{figure}

\subsection{Approximation by Elementary Trigonometry}

Fix any value of $\rho > 0$; for convenience, write $A = A(\rho), B = B(\rho), C = C(\rho)$. We want
an upper bound on the expression inside \cref{eq:2d-gaussian-integral},
\begin{equation}
\label{eq:2d-fixed-rho}
\frac{1}{S^2} \left( \gamma_2(A) + \gamma_2(C) - \gamma_2(B) \right) .
\end{equation}
We will approximate the 2-dimensional Gaussian measure by the plane area; let $\lambda(\cdot)$
denote the plane area measure (\ie Lebesgue measure).  The two-dimensional
Gaussian PDF is $\phi(x) \define \frac{1}{2\pi} e^{-\frac{1}{2} \|x\|_2^2}$.  Since $t < 1$, the
ratio between the Gaussian density at the origin and at the corner $(t,t)$ of the square $[0,t]^2$
is
\[
    \frac{\phi((0,0))}{\phi((t,t))} = e^{t^2} \leq 1 + 2t^2 .
\]
By upper- or lower-bounding the densities in each region by either $\phi((0,0))$ or $\phi((t,t))$,
we get an approximation in terms of the plane area:
\begin{equation}
\label{eq:2d-area-approx}
\begin{aligned}
    \frac{1}{S^2}(\gamma_2(A) + \gamma_2(C) - \gamma_2(B))
    &\leq \frac{1}{\phi((t,t))\lambda([0,t]^2)} \Big(\phi((0,0))\lambda(A) +
        \phi((0,0)) \lambda(C) - \phi((t,t)) \lambda(B) \Big) \\
    &\leq \frac{1}{t^2} \Big( \lambda(A) + \lambda(C) - \lambda(B)\Big)  + 4\lambda(A) ,
\end{aligned}
\end{equation}
where we used $\lambda(C) \leq \lambda(A)$ to get the term on the right-hand side.
Let $\beta = \sin^{-1} \rho$. The $A$, $B$, and $C$ triangles are similar right triangles with one 
angle $\beta$ and corresponding side lengths $\ell_A$, $\ell_B$, and $\ell_C$ respectively (with $\ell_A = t$ and
$\ell_B + \ell_C = t$), so their areas are $\tfrac{1}{2}\ell_A^2 \cdot \frac{\sin\beta}{\cos\beta},$ $\tfrac{1}{2}\ell_B^2 \cdot \frac{\sin\beta}{\cos\beta}$ and $\tfrac{1}{2}\ell_C^2 \cdot \frac{\sin\beta}{\cos\beta}$
 respectively. By symmetry\footnote{More concretely, notice that the
 quadrilateral with angle labeled $\beta$ in \cref{fig:trig} is bisected by the dashed line
 into two congruent right triangles with a shared hypotenuse and a length-$t$
 leg.}
 $C$ is equivalent to the triangle that the ray defined by $w_2$ cuts off from $B$; it follows by elementary trigonometry that (consult \cref{fig:trig}),
\begin{align*}
    \frac{1}{t^2}\left(\lambda(A) + \lambda(C) - \lambda(B)\right)
    &= \frac{1}{2t^2} \cdot \left( \ell_A^2 + \ell_C^2 - \ell_B^2 \right)
        \frac{\sin\beta}{\cos\beta} \\
    &= \frac{1}{2} \left( 1 + \frac{\sin^2(\beta/2)}{\cos^2(\beta/2)}
        - \left(1-\frac{\sin(\beta/2)}{\cos(\beta/2)}\right)^2\right) \frac{\sin\beta}{\cos\beta} \\
    \tag{half-angle identities \&c.}
    &= \frac{1}{\cos\beta} - 1  = \frac{1}{\sqrt{1-\rho^2}} - 1 \\
    &= \frac{1 - \sqrt{1-\rho^2}}{\sqrt{1-\rho^2}}
    \leq \frac{\rho^2}{\sqrt{1-\rho^2}} .
\end{align*}
Substituting back into \cref{eq:2d-area-approx} (with $\lambda(A) = \tfrac{1}{2}t^2 \cdot
\frac{\sin\beta}{\cos\beta} = \tfrac{1}{2} t^2 \tfrac{\rho}{\sqrt{1-\rho^2}}$),
\begin{equation}
\label{eq:2d-final}
    \frac{1}{S^2}\left(\gamma_2(A) + \gamma_2(C) - \gamma_2(B)\right)
        \leq \frac{\rho^2}{\sqrt{1-\rho^2}} + 2 t^2 \frac{\sin\beta}{\cos\beta}
        = \frac{\rho^2}{\sqrt{1-\rho^2}} + 2 t^2 \frac{\rho}{\sqrt{1-\rho^2}} .
\end{equation}

\subsection{Completing the Proof}
Substituting \cref{eq:2d-final} into \cref{eq:2d-gaussian-integral}, we get
\[
    1 + \chi^2(p,q)
    \leq 2 \int_0^1 \left( 1 + \frac{\rho^2}{\sqrt{1-\rho^2}}
        + 2 t^2\frac{\rho}{\sqrt{1-\rho^2}} \right)^m d\sigma_n(\rho)  .
\]
It is now straightforward to bound this quantity using concentration of measure for unit vectors.
We split the integral into the case $\rho < \tfrac{1}{\sqrt 2}$ and $\tfrac{1}{\sqrt 2} < \rho < 1$.
In the first case, since $5t^4m < 1$, we can use $(1+5t^4)^m \leq e^{5t^4m} \leq 1 + 10 t^4m$ to
bound the integral by
\begin{align*}
\int_0^{1/\sqrt{2}} (1 + \sqrt 2 \rho^2 + 2 \sqrt{2} t^2\rho)^m d\sigma(\rho)
&\leq \int_0^{t^2} (1 + 5t^4)^m  d\sigma(\rho)
        + \int_{t^2}^{1/\sqrt{2}} (1 + 5\rho^2)^m d\sigma(\rho) \\
&\leq (1 + 10 t^4m) \int_0^{t^2} d\sigma(\rho)
        + \int_{t^2}^{1/\sqrt{2}} (1 + 5\rho^2)^m d\sigma(\rho) \\
&\leq (1 + 10 t^4m) \int_0^1 (1+5\rho^2)^m d\sigma(\rho) \\
&\leq (1 + 10 t^4m) \cdot \frac{1}{2} \Exu{\bm \rho}{ (1+5\bm \rho^2)^m }  .
\end{align*}
The following is a standard concentration bound for $\bm \rho$ (see \eg Lemma~2.2 of \cite{ball1997elementary}):
\[
    \forall \gamma > 0\colon\qquad \Pr{ |\bm \rho| > \gamma } \leq 2e^{-\gamma^2 n / 2} .
\]
Therefore
\begin{align*}
\Ex{(1+5\bm \rho^2)^m}
&= \int_0^\infty \Pr{ (1+5\bm \rho^2)^m > x } dx 
\leq \int_0^1 dx + \int_0^\infty \Pr{ e^{5\bm \rho^2 m} > 1 + x } dx \\
&\leq 1 + \int_0^\infty \Pr{ |\bm \rho| > \sqrt{\frac{1}{5m}\ln(1+x)} } dx 
\leq 1 + 2\int_0^\infty (1+x)^{-\frac{n}{10 m}} dx
= 1 + \frac{20 m}{n - 10 m} .
\end{align*}
Using $t = n^{-1/4}$, this gives a final bound for the case $0 < \rho < 1/\sqrt 2$ of
\begin{align*}
    2\int_0^{1/\sqrt 2} \left(1 + \frac{\rho^2}{\sqrt{1-\rho^2}} + 2 t^2
\frac{\rho}{\sqrt{1-\rho^2}}\right)^m d\sigma_n(\rho)
    &\leq 2(1 + 10 t^4 m)\cdot \frac{1}{2} \Ex{(1+5 \bm \rho^2)^m} \\
    &\le \left(1 + \frac{10 m}{n}\right)\left(1 + \frac{20 m}{n- 10m}\right) .
\end{align*}

In the case $\tfrac{1}{\sqrt 2} < \rho < 1$, we require the expression for the PDF for the first
coordinate of a random $n$-dimensional unit vector:
\[
    d\sigma_n(\rho) = c_n (1-\rho^2)^{\frac{n-3}{2}} d\rho ,\qquad\text{ where }\qquad
        c_n \define \frac{\Gamma\left(\tfrac{n}{2}\right)}{\sqrt{\pi}
        \Gamma\left(\tfrac{n-1}{2}\right)} \leq \sqrt{\frac{n}{2\pi}} .
\]
Then, for $m < n/3$ and sufficiently large $n$, we can bound the integral by
\begin{align*}
    \int_{1/\sqrt{2}}^1 \left(1 + (\rho + 2t^2)\frac{\rho}{\sqrt{1-\rho^2}} \right)^m d\sigma(\rho)
    &= c_n \int_{1/\sqrt{2}}^1 \left(1 + (\rho + 2t^2)\frac{\rho}{\sqrt{1-\rho^2}} \right)^m
        (1-\rho^2)^{\frac{n-3}{2}} d\rho \\
    &= c_n \int_{1/\sqrt{2}}^1 \left(1 - \rho^2 + (\rho+2t^2)\rho\sqrt{1-\rho^2} \right)^m   
        (1-\rho^2)^{\frac{n-2m-3}{2}} d\rho \\
\intertext{as $\rho\sqrt{1-\rho^2} \leq 1/2$ and $\rho^2 \geq 1/2$ we then have}\quad
    &\le c_n \int_{1/\sqrt{2}}^1 \left(\frac{1}{2} + \frac{1}{2} + t^2\right)^m 2^{-\frac{n-2m-3}{2}} d\rho \\
    &\leq 2^{-\eta n}
\end{align*}
where $\eta > 0$ is some constant
Putting both cases together with \cref{eq:tv-to-chi2} and
\cref{eq:2d-gaussian-integral}, we conclude the proof of
\cref{lemma:continuous-slab-indistinguishability} with
\[
    \TV(p,q) \leq \frac{1}{2} \sqrt{\chi^2(p,q)} \leq \frac{1}{2} \cdot
\sqrt{\left(1+\frac{20m}{n-10m}\right)^2 - 1 + 2^{-\eta n}} = O(\sqrt{m/n}) .
\]

\section{Technicalities for Completing \cref{thm:main}}
\label{sec:proof-of-theorem}
\label{section:technicalities}

To complete the lower bound using \cref{lemma:indistinguishability}, three technical aspects remain.
First, the construction above uses randomized labels in the NO case, but
\cref{def:halfspace-testing} requires the input to be labeled by some function $f$. Second, we need
to establish the dependence on $\epsilon$. Finally, we need to establish the additive
$\tfrac{1}{\epsilon}\log(1/\delta)$ term. These are all handled by standard techniques, and we give
the details here.

\paragraph*{From randomized to deterministic labels.}
By \cref{lemma:yes-halfspace}, any halfspace tester must accept with probability at least $3/4$ when
its input is $\TupleS_\yes$. Define random variable $\TupleS'_\no$ as follows:
\begin{enumerate}
    \item Draw a uniformly random function $\bm{f} : \pmset^{n+\ell} \to \pmset$ and, independently,
        a tuple $\TupleS_\no$.
    \item Output the tuple of labeled samples $(\bz_i, \bm{f}(\bz_i))$ for each $(\bz_i,
        \vec{\bb}_i) \in \TupleS_\no$.
\end{enumerate}
Conditional on no point $\bz$ appearing more than once in the tuple, the distributions of
$\TupleS_\no$ and $\TupleS'_\no$ are identical. By a union bound, the probability that any pair of
the $m < n$ samples collide is at most $2^{-\Omega(n)}$, since the first $\Theta(n)$ bits of each
sample are uniform, so $\dist_\TV(\TupleS_\no, \TupleS'_\no) = o_n(1)$. By
\cref{lemma:indistinguishability} and the triangle inequality, $\dist_\TV(\TupleS_\yes,
\TupleS'_\no) = o_n(1)$. Thus no algorithm that draws $cn$ samples distinguishes between these
inputs with constant advantage.

We now claim that, with probability $1 - o_n(1)$ over the choice of $\bm{f}$ above, samples from
$\TupleS'_\no$ are \iid samples from a distribution $\cD$ such that $\dist_\cD(\bm{f}, h) > \alpha$
for every halfspace $h$, for some small absolute constant $\alpha > 0$. This follows from a counting
argument (\eg \cite[Lemma~2.7]{BFH21}) once we observe that, for any fixed choice of ``key'' vector
$a$, the marginal distribution of $\cD_\no^a$ over the point $\bz$ (denoted by $\cD^a_z$) is uniform
over some support of size $2^{\Omega(n)}$. Indeed this is the case, as for each $x \in \pmset^n$ and
$\abra{k} \in \pmset^{\ell}$, we have
\[
    \cD^a_z(x \circ \abra{k})
    = \Pru{\bx \sim \pmset^{n}}{\bx = x}
        \cdot \Pru{\bb' \sim \pmset}{\floor{a^\top x} + \bb' = k} \\
    = 2^{-n-1} \cdot \ind{k - \floor{a^\top x} \in \pmset} \,.
\]
Therefore any halfspace tester with sufficiently small constant proximity parameter must reject with
probability $3/4 - o_n(1)$ on input $\TupleS'_\no$, and hence must have sample complexity
$\Omega(n)$.

\paragraph*{Dependence on $\epsilon$.}
To recover the dependence on $\epsilon$, it suffices to place the hard input $\TupleS_\yes$ or
$\TupleS'_\no$ in a subcube within a hypercube one dimension larger and assign probability mass
$\epsilon$ to it, while the remaining $1-\epsilon$ mass is placed on a ``dummy'' example in the
other subcube. Concretely, let $n^* \define n + \ell + 1$ and define distributions $\TupleS^*_\yes$
and $\TupleS^*_\no$ as follows. To draw a sample from $\TupleS^*_\yes$,
\begin{enumerate}
    \item Draw a tuple $\TupleS_\yes = ((\bz_i, \vec{\bb}_i))_{i \in [m]}$.
    \item \label{item:transform}
        For each $i \in [m]$, sample a Bernoulli random variable $\bs_i \sim \Bern(\epsilon)$ with
        expectation $\epsilon$. If $\bs_i = 1$, set $(\bz^*_i, \vec{\bb}^*_i) \define ((\bz_i, -1),
        \vec{\bb}_i)$. Otherwise, if $\bs_i = 0$, set $(\bz^*_i, \vec{\bb}^*_i) \define (\vec{1} \in
        \pmset^{n^*}, 1)$.
    \item Output tuple $\TupleS^*_\yes \define ((\bz^*_i, \vec{\bb}^*_i))_{i \in [m]}$.
\end{enumerate}
Define $\TupleS^*_\no$ analogously with respect to $\TupleS'_\no$.

By concentration of measure and the data processing inequality, if $m < cn/\epsilon$ for
sufficiently small $c$, then $\dist_\TV(\TupleS^*_\yes, \TupleS^*_\no) = o_n(1)$ (since with high
probability there are $\le 2cn$ non-dummy samples).

Moreover, samples from $\TupleS^*_\yes$ are \iid samples from a distribution over $\pmset^{n^*}$
labeled according to some halfspace (obtained from the halfspace $h_{\ba}$ from
\cref{lemma:yes-halfspace} by choosing the weight on the last coordinate and the threshold
appropriately). Thus any halfspace tester over $\pmset^{n^*}$ must accept with probability at least
$3/4$ on input $\TupleS^*_\yes$.

Finally, whenever the choice of $\bm{f}$ in the construction of $\TupleS'_\no$ is such that
$\dist_\cD(\bm{f}, h) > \alpha$ for every halfspace $h$, where $\cD = \cD^a_z$ for some vector $a$,
it follows that samples from $\TupleS^*_\no$ are \iid samples from a distribution $\cD^*$ over
$\pmset^{n^*}$ labeled by a function $\bm{f}^*$ satisfying $\dist_{\cD^*}(\bm{f}^*, h^*) > \alpha
\epsilon$ for every halfspace $h^*$ over $\pmset^{n^*}$ (since, from a halfspace $h^*$ closer to
$\bm{f}^*$ over $\cD^*$, we could obtain a halfspace $h$ closer to $\bm{f}$ over $\cD$ by adjusting
the threshold). Thus any halfspace tester over $\pmset^{n^*}$ with proximity parameter $\alpha
\epsilon$ for sufficiently small $\alpha$ must reject with probability $3/4 - o_n(1)$ on input
$\TupleS^*_\no$, and hence must have sample complexity $\Omega(n/\epsilon)$.

\paragraph*{Dependence on $\delta$.}

Thanks to the above proof for the dependence on $\epsilon$, it suffices to show
that $\Omega(\log(1/\delta))$ uniform samples from $\pmset^2$ are required to
distinguish the halfspace $f(x) = x_1$ from the parity function $p(x) = x_1x_2$
with probability at least $1-\delta$. Let $A$ be an algorithm which succeeds in
this task using $m$ samples.
Write $\bm{S}_f$ for the tuple of $m$ labeled samples $(\bx, f(\bx))$ where each $\bx$
is sampled independently and uniformly at random from $\pmset^2$, and similarly $\bm{S}_p$ for
the tuple of samples labeled by $p$.
Then
\begin{align*}
    1-2\delta &\leq
    \Pr{ A(\bm S_f) \text{ outputs } \Accept } - \Pr{ A(\bm S_p) \text{ outputs } \Accept}
    \leq \dist_\TV(\bm S_f, \bm S_p) \\
    &\leq 1 - \frac{1}{2}\left(\Pr{\forall (\bm x, f(\bm x)) \in \bm S_f \;\colon\; \bm x = \vec 1 }
        + \Pr{ \forall (\bm x, p(\bm x)) \in \bm S_{p} \;\colon\; \bm x = \vec 1 } \right) \\
    &\leq 1 - 4^{-m} ,
\end{align*}
so we require $4^{-m} \leq 2\delta$, which concludes the proof.

\end{document}